\documentclass[11pt]{article}
\usepackage[utf8]{inputenc}
\usepackage[english]{babel}
\usepackage{amsmath,amssymb,amsthm}
\usepackage[x11names, rgb]{xcolor}
\usepackage{verbatim}
\usepackage{tikz}
\usepackage{bbold,xspace}

\usepackage{cmap}

\newtheorem{theorem}{Theorem}
\newtheorem{lemma}{Lemma}

\newtheorem{proposition}{Proposition}
\newtheorem{claim}{Claim}

\newcommand{\tw}{{\mathbf{tw}}}

\newcommand{\cO}{\mathcal{O}}

\DeclareMathOperator{\operatorClassNP}{NP}
\newcommand{\classNP}{\ensuremath{\operatorClassNP}}
\DeclareMathOperator{\operatorClassCoNP}{coNP}
\newcommand{\classCoNP}{\ensuremath{\operatorClassCoNP}}
\DeclareMathOperator{\operatorClassFPT}{FPT}
\newcommand{\classFPT}{\ensuremath{\operatorClassFPT}}
\DeclareMathOperator{\operatorClassW}{W}
\newcommand{\classW}[1]{\ensuremath{\operatorClassW[#1]}}

\newcommand{\cost}{W}
\newcommand{\costfunction}{\omega}
\newcommand{\defproblemu}[3]{
  \vspace{1mm}
\noindent\fbox{
  \begin{minipage}{0.95\textwidth}
  #1 \\
  {\bf{Input:}} #2  \\
  {\bf{Question:}} #3
  \end{minipage}
  }
  \vspace{1mm}
}

\begin{document}

\title{Parameterized Complexity of Secluded Connectivity Problems\thanks{The research leading to these results has received funding from the European Research Council under the European Union's Seventh Framework Programme (FP/2007-2013) / ERC Grant Agreement n. 267959 and the Government of the Russian Federation (grant 14.Z50.31.0030).}}

\author{
Fedor V. Fomin\thanks{Department of Informatics, University of Bergen, Norway and St.~Petersburg Department of Steklov Institute of Mathematics of the Russian Academy of Sciences}
\and
Petr A. Golovach\thanks{Department of Informatics, University of Bergen, Norway and St.~Petersburg Department of Steklov Institute of Mathematics of the Russian Academy of Sciences}
\and
Nikolay Karpov\thanks{St.~Petersburg Department of Steklov Institute of Mathematics of the Russian Academy of Sciences}
\and
Alexander S. Kulikov\thanks{St.~Petersburg Department of Steklov Institute of Mathematics of the Russian Academy of Sciences}}

\maketitle
\sloppy

\begin{abstract}
The \textsc{Secluded Path} problem introduced by Chechik et al. in [ESA 2013] models a situation where a sensitive information has to be transmitted between a pair of nodes along a path in a network.
The measure of the quality of a selected path is its \emph{exposure}, which is the total weight of vertices in its closed neighborhood.
 In order to minimize the risk of intercepting the information, we are interested in  selecting a \emph{secluded} path, i.e. a path with a small exposure. Similarly, 
the  \textsc{Secluded Steiner Tree} problem is to find a tree in a graph connecting a given set of terminals such that the  exposure of the tree is 
  minimized.  
  In this work, we obtain the following   
  results about parameterized complexity of secluded connectivity problems.  

We start from an observation that being parameterized by the size of the exposure, the problem is fixed-parameter tractable (\classFPT). More precisely, we give an algorithm deciding if a graph $G$ with a given cost function $\costfunction\colon V(G)\rightarrow \mathbb{N}$ contains a secluded path of exposure at most $k$ with the cost at most $C$ in time $\cO(3^{k/3}\cdot (n+m) \log{\cost})$, where $\cost$ is the maximum value of $\costfunction$ on  an input graph $G$. Similarly, 
\textsc{Secluded Steiner Tree} is  solvable in time $\cO(2^{k}\cdot (n+m) k^2 \log{\cost})$.

The main result of this paper is about  ``above guarantee" parameterizations for secluded problems.  We show that
 \textsc{Secluded Steiner Tree} is 
 \classFPT{}  being  parameterized by $r+p$,  where $p$ is the number of the terminals, $\ell$ the size of an optimum Steiner tree, and $r=k-\ell$. We complement this result by showing that the problem is co-\classW{1}-hard when parameterized by $r$ only.

 We also investigate  
 \textsc{Secluded Steiner Tree} from kernelization perspective and provide several lower and upper bounds when parameters are the treewidth, the size of a vertex cover, maximum vertex degree and the solution size.
Finally, we   refine the algorithmic result of Chechik et al. by improving the exponential dependence from the treewidth of the input graph.
\end{abstract}

\section{Introduction}

\textsc{Secluded Path} and \textsc{Secluded Steiner Tree} problems were introduced in 
Chechik et al. in \cite{ChechikJPP13}. In the \textsc{Secluded Path} problem, 
for  given  vertices $s$ and $t$ of a  graph $G$, the task is to find an $s,t$-path with the minimum \emph{exposure}, i.e.   a path $P$ such that the number of  vertices from $P$ plus the number of vertices of $G$ adjacent to vertices of $P$ is minimized. 
The name secluded comes from the setting where one wants to transfer a confident information over a path in a network which can be intercepted either while passing through a vertex of the path or from some adjacent vertex. Thus the problem is to select a secluded path minimizing the risk of interception of the information.  When instead of connecting two vertices one needs to connect a set of terminals, we arrive naturally to the \textsc{Secluded Steiner Tree}.

More precisely, \textsc{Secluded Steiner Tree} is the following problem.

\defproblemu{\textsc{Secluded Steiner Tree}}{A graph $G$ with a cost function $\costfunction\colon V(G)\rightarrow \mathbb{N}$, a set $S=\{s_1,\ldots,s_p\}\subseteq V(G)$ of terminals, and   non-negative integers $k$ and $C$.}
{Is there a connected subgraph $T$ of $G$ with $S\subseteq V(T)$ such 
                            that $|N_G[V(T)]|\leq k$ and $\costfunction(N_G[V(T)])\leq C$?}

\medskip \noindent
If $\costfunction(v)=1$ for each $v\in V(G)$ and $C=k$, then  we have  an instance of \textsc{Secluded Steiner Tree} without costs; respectively, we omit $\costfunction$ and $C$ whenever we consider such instances.

Clearly, it can be assumed that $T$ is a tree, and thus the problem can be seen as a variant of the classical \textsc{Steiner Tree} problem. 
For the special case $p=2$, we call  the problem \textsc{Secluded Path}.
 
 \noindent\textbf{Previous work.} The study of the secluded connectivity  was initiated by 
Chechik et al. \cite{ChechikJPP12,ChechikJPP13} who showed that 
the decision version of \textsc{Secluded Path} 
without costs
is \classNP-complete. Moreover, for the optimization version of the problem, it is hard to approximate
within a factor of $\cO(2^{\log^{1-\varepsilon}{n}})$,  $n$ is the number of vertices  in the input graph, for any $\varepsilon> 0$ (under an appropriate complexity assumption) \cite{ChechikJPP13}. Chechik et al. \cite{ChechikJPP13} also provided several approximation and parameterized algorithms for \textsc{Secluded Path} and \textsc{Secluded Steiner Tree}. Interestingly, when there are no costs,  
  \textsc{Secluded Path} is solvable in time $\Delta^\Delta \cdot n^{\cO(1)}$, where $\Delta$ is the maximum vertex degree and and thus  is   \classFPT{} being parameterized by $\Delta$.  
Chechik et al. \cite{ChechikJPP13}  also showed that when the treewidth of the input graph does not exceed  $t$, then 
   the \textsc{Secluded Steiner Tree} problem is solvable in time 
$2^{\cO(t\log{t})} \cdot n^{O(1)}\cdot \log{\cost}$,\footnote{In fact, Chechik et al. \cite{ChechikJPP13} give the algorithm that finds a tree with the  exposure of minimum cost, but the algorithm can be easily modified for the more general \textsc{Secluded Steiner Tree}.}
where $\cost$ is the maximum value of $\costfunction$ on  an input graph $G$.
Johnson et al.~\cite{JohnsonLR14} obtained several approximation results for  \textsc{Secluded Path} and showed that the problem with costs is  \classNP{}-hard for subcubic graphs improving the previous result of Chechik et al. \cite{ChechikJPP13} for graphs of maximum degree 4.

The problems related to secluded path and connectivity  under different names 
were considered by several authors. Motivated by 
  secure communications in wireless ad hoc networks, Gao et al.\cite{gao2012thinnest} introduced the very similar notion  of the thinnest path. The motivation of Gilbers \cite{Gilbers13}, who introduced the problem under the name of the minimum witness path, came from the study of art gallery problems.

 \noindent\textbf{Our results.}
In this paper we initiate the systematic study of both problems from the Parameterized Complexity perspective and obtain the following results.
  In Section~\ref{sec:FPT-k}, we show that  \textsc{Secluded Path} and  \textsc{Secluded Steiner Tree} are FPT when parameterized by the size of the solution $k$ by giving algorithms of running time 
$\cO(3^{k/3}\cdot(n+m) \log{\cost})$ and $\cO(2^{k}\cdot (n+m) k^2 \log{\cost})$,  where $\cost$ is the maximum value of $\costfunction$ on  an input graph $G$, correspondingly. 
 
 We  consider the ``above guarantee" parameterizations of both problems in Section~\ref{sec:FPT-above}.
Recall that if $s_1,\ldots,s_p$ are vertices of a graph $G$, then a connected subgraph $T$ of $G$ of minimum size such that $s_1,\ldots,s_p\in V(G)$ is called a \emph{Steiner tree} for the terminals $s_1,\ldots,s_p$. If $p=2$, then a Steiner tree is a shortest $(s_1,s_2)$-path. 
Clearly, if $\ell$ is the size (the number of vertices) of a Steiner tree, then for any connected subgraph $T$ of $G$ with $S\subseteq V(T)$, $|N_G[V(T)]|\geq \ell$. 
Recall that the \textsc{Steiner Tree} problem is well known
to be \classNP-complete as it is included in the famous Karp's list of 21 \classNP-complete problems~\cite{Karp72}, but  in 1971 Dreyfus and Wagner~\cite{DreyfusW71} proved that the problem can be solved in time $O^*(3^p)$, i.e., it is \classFPT ~when parameterized by the number of terminals. The currently best \classFPT-algorithms for \textsc{Steiner Tree} running in time $O^*(2^p)$ are given by Bj{\"{o}}rklund et al.~\cite{BjorklundHKK07} and Nederlof~\cite{Nederlof13} (the first algorithm demands exponential in $p$ space and the latter uses polynomial space).
In Section~\ref{sec:FPT-above} we show that  
\textsc{Secluded Path} and \textsc{Secluded Steiner Tree} are \classFPT{} when the problems are parameterized by $r+p$,  where $r=k-\ell$. From the other side, we
show that the problem is co-\classW{1}-hard when parameterized by $r$ only.

 In Section~\ref{sec:struct}, we  provide a thorough study of the kernelization of the  problem from the structural paramaterization perspective.
We consider   parameterizations by  the treewidth, size of the solution, maximum degree and the size of a vertex cover of the input graph.   We show that
  it is unlikely that \textsc{Secluded Path} (even without costs) parameterized by the solution size, the treewidth and the maximum degree of the input graph,  
admits a polynomial kernel. In particular, this  complements the  \classFPT{} algorithmic findings of Chechik et al. \cite{ChechikJPP13}  for graphs of bounded treewdith and of bounded maximum vertex degree.
  The same holds for the ``above guarantee" parameterization instead the solution size as well.
On the other hand, we show that \textsc{Secluded Steiner Tree} has a polynomial kernel when parameterized by $k$ and the vertex cover number of the input graph. Interestingly, when we parameterize only by the vertex cover number, again, we  show that most likely the problem does not admit a polynomial kernel.
Finally, we refine  the algorithm on graphs of bounded treewidth of  Chechik et al.~\cite{ChechikJPP13} by showing that \textsc{Secluded Steiner Tree}
without costs 
can be solved by a randomized algorithm in time that single-exponentially depends on treewidth by applying the Count \& Color technique of Cygan et al. \cite{cut-and-count} and further observe
that for the general variant of the problem with costs, the same Count \& Color technique can be used as well and also  
a single-exponential deterministic algorithm can be obtained by making use the representative set technique developed by Fomin et al.~\cite{fomin2014efficient}.

\section{Basic definitions and preliminaries}\label{sec:defs}
We consider only finite undirected graphs without loops or multiple
edges. The vertex set of a graph $G$ is denoted by $V(G)$ and  
the edge set  is denoted by $E(G)$. Throughout the paper we typically use $n$ and $m$ to denote the number of vertices and edges respectively.

For a set of vertices $U\subseteq V(G)$,
$G[U]$ denotes the subgraph of $G$ induced by $U$.
For a vertex $v$, we denote by $N_G(v)$ its
\emph{(open) neighborhood}, that is, the set of vertices which are adjacent to $v$, and for a set $U\subseteq V(G)$, $N_G(U)=(\cup_{v\in U}N_G(v))\setminus U$.
The \emph{closed neighborhood} $N_G[v]=N_G(v)\cup \{v\}$. Respectively, 
$N_G[U]=N_G(U)\cup U$. 
For   a set $U\subseteq V(G)$, $G-U$ denotes the subgraph of $G$ induced by $V(G)\setminus U$.
If $U=\{u\}$, we write $G-u$ instead of $G-\{u\}$.
The \emph{degree} of a vertex $v$ is denoted by $d_G(v)=|N_G(v)|$.
We say that a vertex $v$ is \emph{pendant} if $d_G(v)=1$.
 We can omit subscripts if it does not create confusion.
A vertex $v$ of a connected graph $G$ with at least 2 vertices is a \emph{cut vertex} if $G-u$ is disconnected. A connected graph $G$ is \emph{biconnected} if it has at least 2 vertices and has no cut vertices. A \emph{block} of a connected graph $G$ is an inclusion-maximal biconnected subgraph of $G$. A block is \emph{trivial} if it has exactly 2 vertices.  
We say that vertex set $X$ is connected if $G[X]$ is connected. 

A \emph{tree decomposition} of a graph $G$ is a pair $(\mathcal{B},T)$ where $T$
is a tree and $\mathcal{B}=\{B_{i} \mid i\in V(T)\}$ is a collection of subsets (called {\em bags})
of $V(G)$ such that 
\begin{itemize}
\item[i)] $\bigcup_{i \in V(T)} B_{i} = V(G)$, 
\item[ii)] for each edge $xy \in E(G)$, $x,y\in B_i$ for some  $i\in V(T)$, and 
\item[iii)] for each $x\in V(G)$ the set $\{ i \mid x \in B_{i} \}$ induces a connected subtree of $T$.
\end{itemize}
The \emph{width} of a tree decomposition $(\{ B_{i} \mid i \in V(T) \},T)$ is $\max_{i \in V(T)}\,\{|B_{i}| - 1\}$. The \emph{treewidth} of a graph $G$ (denoted as $\tw(G)$) is the minimum width over all tree decompositions of $G$. 

A set $U\subseteq V(G)$ is a \emph{vertex cover} of $G$ if for any edge $uv$ of $G$, $u\in U$ or $v\in U$. The \emph{vertex cover number} $\tau(G)$ is the size of a minimum vertex cover.

Parameterized complexity is a two dimensional framework
for studying the computational complexity of a problem. One dimension is the input size
$n$ and another one is a parameter $k$. It is said that a problem is \emph{fixed parameter tractable} (or \classFPT), if it can be solved in time $f(k)\cdot n^{O(1)}$ for some function $f$. 
A \emph{kernelization} for a parameterized problem is a polynomial algorithm that maps each instance $(x,k)$ with the input $x$ and the parameter $k$ to an instance $(x',k')$ such that i) $(x,k)$ is a yes-instance if and only if $(x',k')$ is a yes-instance of the problem, and ii) the size of $x'$ is bounded by $f(k)$ for a computable function $f$. 
The output $(x',k')$ is called a \emph{kernel}. The function $f$ is said to be a \emph{size} of a kernel. Respectively, a kernel is \emph{polynomial} if $f$ is polynomial. 
While a parameterized problem is \classFPT{} if and only if it has a kernel, it is widely believed that not all \classFPT{} problems have polynomial kernels. In particular, Bodlaender et al.~\cite{BodlaenderDFH09,BodlaenderJK14} introduced techniques that allow to show that a parameterized problem has no polynomial kernel unless  $\classNP\subseteq\classCoNP/\text{\rm poly}$.
We refer to the book of Downey and Fellows~\cite{DowneyF13},
 for  detailed introductions  to parameterized complexity. 

We use randomized algorithms for our problems. Recall that a \emph{Monte Carlo algorithm} is a randomized algorithm whose running time is deterministic, but whose output may be incorrect with a certain (typically small) probability. A Monte-Carlo algorithm is \emph{true-biased} (\emph{false-biased} respectively) if it always returns a correct answer when it returns a yes-answer (a no-answer respectively).

\section{\classFPT-algorithms for the problems parameterized by the solution size}\label{sec:FPT-k}

In this section we consider \textsc{Secluded Path} and \textsc{Secluded Steiner Tree} problems parameterized by the size of the solution, i.e., by $k$. We also show how these parameterized algorithms can be used to design faster exact exponential algorithms. 

We start with \textsc{Secluded Path}.

\begin{theorem}\label{lemma:normenum}
\textsc{Secluded Path} is  solvable in time 
$\cO(3^{k/3}\cdot n \log{\cost})$,
where $\cost$ is the maximum value of $\costfunction$ on  an input graph $G$.
\end{theorem}
\begin{proof}
Let us observe first that if there is an optimal secluded path, then there is an optimal secluded induced path---shortcutting a path cannot increase the size of its neighbourhood.
 We give an algorithm that enumerates all induced paths $P$ from $u$ to $v$ such that $|N_G[V(P)]|\leq k$ in time 
$O(3^{k/3}\cdot n)$ for a graph $G$ with $n$ vertices.
 Then picking up a secluded path  of minimum cost will complete the proof. 

The algorithm is based on the standard branching ideas.
If $|N_G[u]|>k$ the algorithm reports that no such path exist and stops.
If $|N_G[u]|\leq k$ and $u=v$ the algorithm outputs the path consisting of the single vertex~$u$. Otherwise a path from $u$ to $v$ must go through one of the neighbors of~$u$. Since we are looking for an induced path it must never return to a vertex from $N_G[u]$. This allows us to branch as follows. For each $w \in N_G(u)$, 
we check recursively whether the graph $G_w=(G \setminus N_G[u]) \cup \{w\}$ contains
an induced path $Q$ from $w$ to $v$ such that $|N_{G_w}[Q]|\leq k-|N_G(u)|$. This way we get the following recurrence on the number of nodes $t(k)$ in the corresponding recursion tree. If $u=v$, then there is only one path from $u$ to $v$, and $t(k)\leq 1$. If $u\neq v$,  then 
$t(k) \le  d\cdot t(k-d)$,
where $d=|N_G(u)|$.
This is a well known recurrence implying that $t(k)=\cO(3^{k/3})$
(see, e.g., the analysis of the algorithm enumerating all maximal independent sets
in Chapter~1 of \cite{FominKratsch2010}). 
Note that we spend only a linear time $O(n)$ in each vertex of the recursion tree.
Since the length of each path $P$ can be computed in time $\cO(n\log\cost)$, we can find a path of minimum cost it time $\cO(3^{k/3}n\log\cost)$. Therefore, the total running time is $\cO(3^{k/3}\cdot n \log{\cost})$.
\end{proof}

 For \textsc{Secluded Steiner Tree} we prove the following theorem.

\begin{theorem}\label{thm:sec_tree}
\textsc{Secluded Steiner Tree} can be solved in time $\cO(2^{k}\cdot (n+m) k^2 \log{\cost})$,
where $\cost$ is the maximum value of $\costfunction$ on  an input graph $G$.
\end{theorem}

The following proposition from \cite{FominV12} will be useful for us.

\begin{proposition}[\cite{FominV12}]\label{le:connectedComp} Let $G$ be a graph. For every $v\in
V(G)$, and $b,f\geq 0$,  the number of connected vertex subsets
$B\subseteq V(G)$ such that
\begin{itemize}
\item[(i)]
 $v \in B$,\item[(ii)] $|B| = b+1$, and \item[(iii)] $|N_G(B)|=f$,  \end{itemize}
  is at most $\binom{b+f}{b}$.
 Moreover, all such subsets can be enumerated in time $\cO (\binom{b+f}{b}\cdot (n+m) \cdot b\cdot ( b+ f))$.
\end{proposition}

\begin{proof}[Proof of Theorem~\ref{thm:sec_tree}]
By Proposition~\ref{le:connectedComp}, the number of  connected sets $T$ of size $b$ 
containing $s_1$ and such that $|N_G[T]|=b+f$, does not exceed  
$\binom{b+f}{b}$. Since $b+f\leq k$ and there are at most $k^2$ choices for the values of $b$ and $f$, we have that the number of such sets does not exceed
 $\binom{b+f}{b} k^2$.
By Proposition~\ref{le:connectedComp}, all such sets $N_G[T]$ can be enumerated in time $2^{k}\cdot   (n+m)) k^2$.  While 
  enumerating  sets $N_G[T]$, we disregard  sets not containing all terminal vertices. Finally, we select the set of minimum cost. 
 \end{proof}

Parameterized algorithms for  \textsc{Secluded Path} and \textsc{Secluded Steiner Tree}  combined with a brute-force procedure imply the following exact exponential algorithms for the problems. 

\begin{theorem}\label{thm:exact_secluded} On an $n$-vertex graph, 
\textsc{Secluded Path} is solvable in time $\cO(1.3896^n\cdot\log\cost)$ and 
\textsc{Secluded Steiner Tree} is solvable  in time $\cO(1.7088^n\cdot \log\cost)$,
where $\cost$ is the maximum value of $\costfunction$ on  an input graph $G$.
\end{theorem}

\begin{proof}
By Theorem~\ref{lemma:normenum},  \textsc{Secluded Path}  is solvable in time $3^{k/3}\cdot n\log\cost$. 
 On the other hand,  we also can solve the problem 
  by the  brute-force procedure checking 
for every set $X$ of size  $n-i$, 
whether $V(G)\setminus X $ and $\costfunction(V(G)\setminus X)\leq C$.
Notice that  $V(G)\setminus X$  contains the closed neighborhood of a secluded path if and only if  $V(G)\setminus N_G[X]$  is connected and contains both  terminal vertices, and these conditions can be checked in polynomial time.
The brute-force procedure takes  time $\binom{n}{n-i} \cdot  n^{\cO(1)}\log\cost$.

Let us note that for $\varepsilon\geq 0.8983$, $3^{\varepsilon n/3 }>\binom{n}{(1-\varepsilon)n }$.
Thus  for all integers $i$ between  $0.8983 \cdot n$ and $n$, we enumerate sets of size $n-i$, while for all integers $i$ between  $1$ and $0.8983 \cdot n$ we use Theorem~\ref{lemma:normenum}  to find if there is a solution of size at most $i$. The running time of the algorithm is dominated by 
$\cO(3^{\frac{0.8983 n}{3}}\cdot\log\cost)=\cO(1.3896^n\cdot\log\cost)$.


Similarly, we use parameterized  time   $2^{k}\cdot n^{\cO(1)}\log\cost$ algorithm from Theorem~\ref{thm:sec_tree} for   \textsc{Secluded Steiner Tree} and balance it with  
 the  brute-force procedure checking 
for every set $X$ of size  $n-i$, 
whether $V(G)\setminus X $ is the closed neighbourhood of a secluded Steiner tree $T$. 
For each such set $X$, we  check in polynomial time 
whether $V(G)\setminus N_G[X]$  is connected and contains all terminal vertices.
The brute-force runs in time $\binom{n}{n-i} \cdot  n^{\cO(1)}\log\cost$.

For $\varepsilon\geq 0.77923$, we have that  $2^{\varepsilon n }>\binom{n}{(1-\varepsilon)n}$. Thus  for all integers $i$ between  $0.77923\cdot n$ and $n$, we enumerate sets of size $n-i$, while for all integers $i$ between  $1$ and $0.77923\cdot n$ we use Theorem~\ref{thm:sec_tree} to find if there is a solution of size at most $i$. The running time of this algorithm is  $\cO(2^{0.77923 n}\cdot \log\cost)=\cO(1.7088^n\cdot\log\cost)$.
\end{proof}

\section{\classFPT-algorithms for the problems parameterized above the guaranteed value}\label{sec:FPT-above}
In this section we show that \textsc{Secluded Path} and \textsc{Secluded Steiner Tree} are \classFPT{} when the problems are parameterized by $r+p$
where $r=k-\ell$ and $\ell$ is the size of a Steiner tree for $S$.

\begin{theorem}\label{lemma:pathaboveguarantee}
\textsc{Secluded Path} is  solvable in time $\cO(2^{k-\ell}\cdot (n+m) \log{\cost})$, where $\ell$ is the length of a shortest $(u,v)$-path for $\{u,v\}=S$ and $\cost$ is the maximum value of $\costfunction$ on  an input graph $G$.
\end{theorem}

\begin{proof}
The proof of this theorem is very similar to the proof of Theorem~\ref{lemma:normenum}. For an integer $h$, 
we  enumerate in  the graph $G$ 
  all induced paths $P$ from $u$ to $v$ of length at most $h-1$ such that $|N_G(V(P))|\leq k-h$. The only difference with Theorem~\ref{lemma:normenum} is that this time  we bound the running time of the algorithm as a function of $k-h$.

If $|N_G(u)|>k-h$ the algorithm reports that no such path exist and stops.
If $|N_G[u]|\leq  k-h$ and $u=v$ the algorithm outputs the path consisting of the single vertex~$u$. Otherwise, we  branch by checking recursively for each $w \in N_G(u)$, 
  whether the graph $G_w=(G \setminus N_G[u]) \cup \{w\}$ contains
an induced path $Q$ from $w$ to $v$ of length at most $h-1$ such that $|N_{G_w}(Q)|\leq k-|N_G(u)| - (h-1)$. This way we get the following recurrence on the number of nodes $T(k-h)$ in the corresponding recursion tree. If $u=v$, then there is only one path from $u$ to $v$, and $T(k-h)\leq 1$. If $u\neq v$,  then 
\[T(k-h) \le  d\cdot T(k-d-h+1) \, ,\]
where $d=|N_G(u)|$.
It is easy to show, that 
  $T(k-h)=\cO(2^{k-h})$.
\end{proof}

We need some structural properties of solutions of \textsc{Secluded Steiner Tree}.
We start with an auxiliary lemma bounding the number of vertices of degree at least three in $F$ as well as the number of their neighbors. 

\begin{lemma}\label{lem:big-degree}
Let $G$ be a connected graph and $S\subseteq V(G)$, $p=|S|$. Let $F$ be an inclusion minimal induced subgraph of $G$ such that $S\subseteq V(F)$ and $X=\{v\in V(F)|d_F(v)\geq 3\}\cup S$. Then
\begin{itemize}
\item[i)] $|X|\leq 4p-6$, and  
\item[ii)] $|N_F(X)|\leq 4p-6$.
\end{itemize}
 \end{lemma}

\begin{proof}
Let $\mathcal{B}$ be the set of blocks of $F$. Consider  bipartite graph $T$ with the bipartition $(V(F),\mathcal{B})$ of the vertex such that $v\in V(F)$ and $b\in\mathcal{B}$ are adjacent if and only if $v$ is a vertex of $b$. Notice that $T$ is a tree.
Recall that the \emph{vertex dissolution} operation for a vertex $v$ of degree 2 deletes $v$ together with incident edges and replaces them by the edge joining the neighbors of $v$. 
 Denote by $T'$ the tree obtained from $T$ by consequent
  dissolving all  vertices of $T$ of degree 2 that are not in $S$.
Denote by $L$ the set of leaves of $T$. By the minimality of $F$, $L\subseteq S$. Let $q_1=|L|\leq p$, and let $q_2$ be the number of degree 2 vertices and $q_3$ be the number of vertices of degree at least 3 in $T$. Clearly, $q_1+2q_2+3q_3\leq 2|E(T)|=2(q_1+q_2+q_3-1)$. Then $q_3\leq q_1-2\leq p-2$. 
We have that $|\{v\in V(T)|d_T(v)\geq 3\}\cup S|\leq q_3+p\leq 2p-2$ and $|V(T')|\leq 2p-2$.
Observe that if $d_F(v)\geq 3$ for $v\in V(F)\setminus S$, then $v$ is a cut vertex of $F$ and either $v$ is included in at least 3 blocks of $F$, or $v$ is in a block of size at least 3. In the second case, $v$ is adjacent to a vertex $b\in\mathcal{B}$ of $T$ with degree at least 3. It implies that
$|X|\leq 2|E(T')|=2(|V(T')|-1)\leq 4p-6$ and we have (i). To show (ii), observe that 
$|N_F(X)|\leq 2|E(T')|\leq 4p-6$.  
\end{proof}

The following lemma provides a bound on  the number of vertices of a tree  
that have neighbors outside the tree.

\begin{lemma}\label{lem:adj}
Let $G$ be a connected graph and $S\subseteq V(G)$, $p=|S|$. Let $\ell$ be the size of a Steiner tree for $S$ and $r$ be a positive integer. Suppose that $T$ is an inclusion minimal subgraph of $G$ such that $T$ is a tree spanning $S$ and $|N_G[V(T)]|\leq \ell+r$.
Then for $Y=N_G(V(T))$, $|N_G(Y)\cap V(T)|\leq 4p+2r-5$.
\end{lemma}

\begin{proof} Denote by $L$ the set of leaves of $T$ and by $D$ the set of vertices of degree at least 3 in $T$. Clearly, $L\subseteq S$. 
We select a leaf $z$ of $T$ as the \emph{root} of $T$. The selection of a root defines a parent-child relation on $T$. For each $u\in Y$, 
denote by $x(u)$  the vertex in $N_G(u)\cap V(T)$ at minimum distance to $z$ in $T$. Let $U=\{x(u)|u\in Y\}$. For a vertex $u\in Y$ and $v\in N_G(u)\cap V(T)\setminus\{x(u)\}$, let $y(u,v)$
be the parent of $v$ in $T$. Let $W=\{y(u,v)|u\in Y, v\in N_G(u)\cap V(T), v\neq x(u)\}$ and $W'=W\setminus (S\cup D\cup U)$.  

Let $F=G[V(T)\cup Y]$.
\begin{claim}
Set  $F'=F-W'$ is connected. 
\end{claim}
\begin{proof}[Proof of the claim] Since all leaves of $T$ including $z$ are in $S$, we have that 
  $z\in V(F')$.
To prove the claim, we show that for each vertex $v\in V(F')$, there is a $(v,z)$-path in $F'$.  Every vertex $u\in Y$ has a neighbor $x(u)$ in $F'$. Hence, it is sufficient to prove the existence of  $(v,z)$-paths for $v\in V(T)\setminus W'$. The proof is by  induction on the distance between $z$ and $v$ in $T$. If $v=z$, then we have a trivial $(z,v)$-path. Assume that $v\neq z$. Let $w$ be the parent of $v$ in $T$. If $w\in V(F')$, then by the inductive hypothesis, there is a $(z,w)$-path in $F'$ and it implies the existence of a $(z,v)$-path. Suppose that $w\notin V(F')$, i.e., $w\in W'$. Since $d_T(w)=2$, there is $u\in Y$ such that $w=y(u,v)$. The distance in $T$ between $z$ and $x(u)$ is less than the distance between $z$ and $v$. Therefore, by the inductive hypothesis, there is a $(z,x(u))$-path in $F'$. It remains to observe that because $x(u)u,uv\in E(F')$, $F'$ has a $(z,v)$-path as well.
This concludes the proof to the claim. \end{proof}

Denote by $C$ the set of the children of the vertices of $D\cup S$ in $T$. Observe that $|N_G(Y)\cap V(T)|\leq |D\cup S|+|C|+|U|+|W'|$.
Recall that $|V(F)|\leq \ell+r$. Because $F'$ is connected and $S\subseteq V(F')$, $|V(F')|\geq \ell$. Hence, $|W'|\leq r$.  
Let $q_1=|L|$, $q_2=|V(T)\setminus (L\cup D)|$ and $q_3=|D|$.  We have that $q_1+2q_2+3q_3\leq 2|E(T)|=2(q_1+q_2+q_3-1)$. Then $q_3\leq q_1-2$ and $|D\cup S|\leq 2|S|-2=2p-2$, because $L\subseteq S$. Let $T'$ be the tree obtained from $T$ by consequent dissolving all 
the vertices of degree 2 that are not in $S$. Then $|C|\leq |E(T')|\leq 2|S|-3=2p-3$. Since $|V(T)|\geq \ell$, $|U|\leq |Y|\leq r$. We obtain that 
$|N_G(Y)\cap V(T)|\leq |D\cup S|+|C|+|U|+|W'|\leq 2p-2+2p-3+r+r=4p+2r-5$.
\end{proof}

Now we are ready to prove the main result of the section.

\begin{theorem}\label{thm:above-tree}
 \textsc{Secluded Steiner Tree} can be solved in time $2^{O(p+r)}\cdot nm\cdot \log\cost$  by a true-biased Monte-Carlo algorithm and in time $2^{O(p+r)}\cdot nm \log n\cdot \log\cost$ by a deterministic algorithm for graphs with $n$ vertices and $m$ edges,  where $r=k-\ell$ and $\ell$ is the size of a Steiner tree for $S$
and $\cost$ is the maximum value of $\costfunction$ on  an input graph $G$.
\end{theorem}

\begin{proof}
We construct an \classFPT-algorithm for \textsc{Secluded Steiner Tree} parameterized by $p+r$.
The algorithm is based  on the random separation techniques introduced by Cai, Chan, and Chan~\cite{CaiCC06} (see also~\cite{AlonYZ95}). We first describe a randomized algorithm and then explain how it can be derandomized.

Let $\mathcal{I}=(G,\costfunction,S,k,C)$ be an instance of \textsc{Secluded Steiner Tree}, $\ell$ be the size of a Steiner tree for $S=\{s_1,\ldots,s_p\}$ and $r=k-\ell$. Without loss of generality we assume that $p\geq 2$ and $r\geq 1$ as for $p=1$ or $r=0$, the problem is trivial. We also can assume that $G$ is connected.

\paragraph{Description of the algorithm}
In each iteration of the algorithm
we color the vertices of $G$ independently and uniformly at random by two colors. In other words, we partition $V(G)$ into two sets $R$ and $B$. We say that the vertices of $R$ are \emph{red}, and the vertices of $B$ are \emph{blue}. Our algorithm can recolor some blue vertices red, i.e., the sets $R$ and $B$ can be modified.
Our aim is  to find a connected subgraph $T$ of $G$ with $S\subseteq V(T)$ such that $|N_G[V(T)]|\leq k$, $\costfunction(N_G[V(T)])\leq C$ and $V(T)\subseteq R$.

\medskip
\noindent
{\bf Step 1.} If $G[R]$ has a component $H$ such that $S\subseteq V(H)$, then find a spanning tree $T$ of $H$. If $|N_G[V(T)]|\leq k$ and $\costfunction(N_G[V(T)])\leq C$, then return $T$ and stop; otherwise, return  that 
$\mathcal{I}$
is no-instance   and stop.

\medskip
\noindent
{\bf Step 2.} If there is $s_i\in S$ such that $s_i\notin R$ or $N_G(s_i)\cap R=\emptyset$, then return  that 
$\mathcal{I}$
is no-instance  and stop.

\medskip
\noindent
{\bf Step 3.} Find a component $H$ of $G[R]$ with $s_1\in V(H)$. 
If there is a pendant vertex $u\notin S$ of $H$ that is adjacent in $G$ to the unique vertex $v\in B$, then find a component of $G[B]$ that contains $v$, recolor its vertices red and then return to Step~1. 
Otherwise, return that $(G,S,k)$ is no-instance 
  and stop.
\medskip

We repeat at most $2^{O(r+p)}$ iterations. If on some iteration we obtain a yes-answer, then we return it and the corresponding solution. Otherwise, if on every iteration we get a no-answer, we return a no-answer.  

\paragraph{Correctness of the algorithm}
It is straightforward to see that if this algorithm returns a tree $T$ in $G$ with  $|N_G[V(T)]|\leq k$ and $\costfunction(N_G[V(T)])\leq C$, then we have a solution for the considered instance of \textsc{Secluded Steiner Tree}. We show that if 
$\mathcal{I}$
is a yes-instance, then there is a positive constant $\alpha$ that does not depend on $n$ and $r$ such that the algorithm finds a tree $T$ in $G$ with  $|N_G[V(T)]|\leq k$ and $\costfunction(N_G[V(T)])\leq C$ with probability at least $\alpha$ after $2^{O(p+r)}$ executions of  this algorithm for random colorings. 

Suppose that 
$\mathcal{I}$
is a yes-instance. Then there is a tree $T$ in $G$ such that $S\subseteq V(T)$, $|N_G[V(T)]|\leq k$ and $\costfunction(N_G[V(T)])\leq C$. Without loss of generality we assume that $T$ is inclusion minimal. Let $F=G[V(T)]$, $X=\{v\in V(F)|d_F(v)\geq 3\}\cup S$, $X'=N_F(X)$, $Y=N_G(V(T))$ and $Y'=N_G(Y)\cap V(T)$. For each $v\in Y'\setminus S$, 
we arbitrarily select two distinct neighbors 
$z_1(v)$ and $z_2(v)$ in $T$. Because the leaves of $T$ are 
in $S$, we have that $v$ is not a leaf and thus has at least two neighbors.  Let $Z=\{z_i(v)|v\in Y'\setminus S,i=1,2\}$.  Let $W=X\cup X'\cup Y\cup Y'\cup Z$. 

By Lemma~\ref{lem:big-degree}, $|X|\leq  4p-6$ and $|X'|\leq  4p-6$. By Lemma~\ref{lem:adj}, $|Y'|\leq 4p+2r-5$ and, therefore, $|Z|\leq 8p+4r-10$. Because $|V(T)|\geq\ell$ and 
$|N_G[V(T)]|\leq \ell+r$, we have that $|Y|\leq r$. Hence  $|W|\leq |X|+|X'|+|Y|+|Y'|+|Z|\leq 4p-6+4p-6+r+4p+2r-5+8p+4r-10=20p+7r-27$. 
Let  $N=20p+7r-27$. 
Then with probability at least
$2^{-N}$, the vertices of $Y$ are colored blue and the vertices of $X\cup X'\cup Y'\cup Z$ are colored red, i.e., $W\cap V(T)\subseteq R$ and $W\setminus V(T)\subseteq B$.
The probability that for a random coloring, the vertices of $W$ are colored incorrectly, i.e., $W\cap V(T)\cap B\neq\emptyset$ or $(W\setminus V(T))\cap R\neq\emptyset$, is at most 
$1-2^{-N}$. 
Hence, if we consider $2^N$ random colorings, 
then the probability that the vertices of $W$ are colored incorrectly for all the colorings is at most $(1-2^{-N})^{2^N}$, and
with probability at least $1-(1-2^{-N})^{2^N}$ for at least one coloring we will have  $W\cap V(T)\subseteq R$ and $W\setminus V(T)\subseteq B$.   
Since $(1-2^{-N})^{2^N}\leq 1/e$, we have that $1-(1-2^{-N})^{2^N}\leq 1-1/e$.
Thus if $\mathcal{I}$
is a yes-instance, after  $2^N$  random colorings of $G$, we have  that at least one of the colorings is successful with a 
constant success probability
$\alpha=1-1/e$.

Assume that for a random red-blue coloring of $G$, $W\cap V(T)\subseteq R$ and $W\setminus V(T)\subseteq B$. We show that in this case the algorithm finds a tree $T'$ with $S\subseteq V(T')\subseteq V(T)$. Clearly, $|N_G[V(T')]|\leq |N_G[V(T)]|\leq k$ and $\costfunction(N_G[V(T')])\leq \costfunction(N_G[V(T)])\leq C$  in this case.

We claim that for every connected  component  $H$ of $G[R]$, either $V(H)\subseteq V(T)$ or $V(H)\cap V(T)=\emptyset$. To obtain a contradiction, assume that there are $u,v\in V(H)$ such that $u\in V(T)$ and $v\notin V(T)$.
Indeed, $H$ is connected, and thus contains an $(u,v)$ path $P$. Since $P$ goes from $V(T)$ to $v\not\in V(T)$, path $P$ should contain a vertex $w\in N_G(T) =Y$.  But $w$ is colored blue, which is a contradiction to the assumption that $P$ is in the red component $H$. 
  By the same arguments, for any component $H$ of $G[B]$, either $V(H)\subseteq V(T)$ or $V(H)\cap V(T)=\emptyset$.

We consider Steps~1--3 of the algorithm and show their correctness.

Suppose that $G[R]$ has a component $H$ such that $S\subseteq V(H)$. Because $S\subseteq W$ and $S\subseteq V(T)$, $V(H)\subseteq V(T)$. 
Then for  every spanning tree $T'$ of $H$, $S\subseteq V(T')$ and $N_G[V(T')]\subseteq N_G[V(T)]$. 
Therefore, $|N_G[V(T')]|\leq |N_G[V(T)]|\leq k$ and $\costfunction(N_G[V(T')])\leq \costfunction(N_G[V(T)])\leq C$. Hence, if a component of $G[R]$ contains $S$, then we find a solution. 
This concludes the proof of the correctness of the first step. 

Let us assume that the algorithm does not stop at Step~1. For the right coloring,  
because $S\subseteq X$ and $N_F(S)\subseteq X'$, for  every   $s_i\in S$, we have that $s_i\in R$. 
 Moreover, because $p\geq 2$,
  at least one neighbor of $s_i$ in $G$ is in $R$.  Thus the only reason why the algorithm stops at Step~2 is due to the wrong coloring. 
  Consider the case when the algorithm does not stop after Step~2.

Suppose that $H$ is a component of $G[R]$ with $s_1\in V(H)$. Because the algorithm did not stop in Step~2, such a component $H$  exists and has at least 2 vertices.  
Recall that $V(H)\subseteq V(T)$. 
Because we proceed in Step~1, we conclude that $S\setminus V(H)\neq \emptyset$. Then there is a vertex $u\in V(H)$ which has a neighbor $v$ in $T$ such that $v\in B$. If $u\in S$, then $v\in X'$, but this contradicts the assumption  $X'\subseteq R$. Hence, $u\notin S$.
Suppose that 
$d_H(u)\geq 2$.
In this case $d_F(u)\geq 3$ and $v\in X'$; a contradiction. 
Therefore, $u$ is a pendant vertex of $H$. 

Let $u\notin S$ be an arbitrary pendant vertex of $H$. 
If $u$ has no neighbors in $B$, then $u$ is a leaf of $T$ that does not belong to $S$ but this contradicts the inclusion minimality of $T$. 
Assume that $u$ is adjacent to at least two distinct vertices of $B$. 
Because $T$ is an inclusion minimal tree spanning  $S$, vertex $u$ has at least two neighbors in $T$ and $u$ has a neighbor $v\in B$ in $T$. Let $w\in (N_G(u)\cap B)\setminus\{v\}$. 
If $w\in V(T)$, then $d_F(u)\geq 3$ and, therefore, $u\in X$ and $v,w\in X'$; a contradiction with $X'\subseteq R$. \
Hence, $w\notin V(T)$.
Moreover, $v$ is the unique neighbor of $u$ in $T$ that belongs to $B$.  
Then $w\in Y$ and $v\in\{z_1(u),z_2(u)\}$;  a contradiction with $Z\subseteq R$.  
We obtain that  $u$ is adjacent in $G$ to the unique vertex $v\in B$. Let $H'$ be the component of $G[B]$ that contains $v$. Since $T$ is an inclusion minimal tree that spans $S$, $u$ has at least two neighbors in $T$. It implies that $v\in V(T)$,  therefore $V(H')\subseteq V(T)$. We recolor the vertices of $H'$ red in Step~3. For the new coloring the vertices of $Y$ are blue and the vertices of $W\setminus Y$ are red. Therefore, we keep the crucial property of the considered coloring but we increase the size of the component of $G[R]$  containing $s_1$.  

To conclude  the correctness proof, it remains to observe that in Step~3 we increase the number of vertices in the component of $G[R]$ that contains $s_1$. Hence, after at most $n$ iterations, we obtain a component in $G[R]$ that includes $S$ and return a solution in Step~1.

It is straightforward to verify that each of Steps~1--3 can be done in time $O(m\log\cost)$. Because the number of iterations is at most $n$, we obtain that the total running time is 
$2^{O(p+r)}\cdot nm\log\cost$. 

This algorithm can be derandomized by standard techniques (see~ \cite{AlonYZ95,CaiCC06}). The random colorings can be replaced by the colorings induced by \emph{universal sets}.
Let $n$ and $q$ be positive integers, $q \leq n$. An  \emph{$(n,q)$-universal set} is a collection of binary vectors of length $n$ such that for each index subset of size $q$, each of the $2^q$ possible combinations of values appears in some vector of the set. It is known that an $(n,q)$-universal set can be constructed in \classFPT-time with the parameter $q$. The best construction is due to Naor, Schulman and Srinivasan~\cite{naor1995splitters}. They obtained an $(n,q)$-universal set of size $2^q\cdot q^{O(\log q)} \log n$, and proved that the elements of the sets  can be listed in time that is linear in the size of the set. In our case $n$ is the number of vertices of $G$ and $q=20p+7r-27$. 
\end{proof}

We complement Theorem~\ref{thm:above-tree} by showing that it is unlikely that \textsc{Secluded Steiner Tree} is \classFPT\ if parameterized by $r$ only. To show it, we use the standard reduction from the \textsc{Set Cover} problem (see, e.g., \cite{Karp72}). Notice that we prove that \textsc{Secluded Steiner Tree} is co-\classW{1}-hard, i.e., we show that it is \classW{1}-hard to decide whether we have a no-answer.

\begin{theorem}\label{thm:co-W-hard}
\textsc{Secluded Steiner Tree} without costs is co-\classW{1}-hard when parameterized by $r$,
 where $r=k-\ell$ and $\ell$ is the size of a Steiner tree for $S$.
\end{theorem}

\begin{proof}[of Theorem~\ref{thm:co-W-hard}]
Recall that the \textsc{Set Cover} problem for a set $U$, subsets $X_1,\ldots,X_m\subseteq X$ and a positive integer $k$, asks whether there are $k'\leq k$ sets $X_{i_1},\ldots,X_{i_{k'}}$ for $i_1,\ldots,i_{k'}\in\{1,\ldots,m\}$
that \emph{cover} $U$, i.e., $U\subseteq \cup_{j=1}^kX_{i_j}$. As it was observed in~\cite{GutinJY11}\footnote{Gutin et al. prove in \cite{GutinJY11} the statement for the dual \textsc{Hitting Set} problem.}, \textsc{Set Cover} is \classW{1}-hard when parameterized by $p=m-k$. 
To prove the theorem, we reduce this parameterized variant of \textsc{Set Cover}. 

Let $(U,X_1,\ldots,X_m,k)$ be an instance of \textsc{Set Cover}. Let $U=\{u_1,\ldots,u_n\}$. We construct the bipartite graph $G$ as follows.
\begin{itemize}
\item[i)] Construct $m$ vertices $x_1,\ldots,x_m$ and $n$ vertices $u_1,\ldots,u_n$.
\item[ii)] For $i\in\{1,\ldots,m\}$ and $j\in \{1,\ldots,n\}$, construct an edge $x_iu_j$ if $u_j\in X_i$.
\item[iii)] Construct a vertex $y$ and join it with $x_1,\ldots,x_m$ by edges. 
\end{itemize}
Let $S=\{y,u_1,\ldots,u_n\}$ and $r=m-k-1$.

Suppose that $(U,X_1,\ldots,X_m,k)$ is a yes-instance of \textsc{Set Cover} and  assume that $X_{i_1},\ldots,X_{i_{k'}}$ cover $U$. Then $F=G[S\cup\{x_{i_1},\ldots,x_{i_{k'}}\}]$ is a connected subgraph of $G$ and $S\subseteq V(F)$. Clearly, $|V(F)|\leq n+k+1$. Let $T$ be a Steiner tree for the set of terminals $S$. We have that $\ell=|V(T)|\leq |V(F)|\leq n+k+1$. 
Notice that for any connected subgraph $T'$ of $G$ such that $S\subseteq V(T')$, $N_G[V(T')]=V(G)$. 
We have that for any connected subgraph $T'$ of $G$ with $S\subseteq V(G)$, $|N_G[V(T')]|=n+m+1>(n+k+1)+(m-k-1)\geq\ell+r$. Therefore, $(G,S,\ell+r)$ is a no-instance of \textsc{Secluded Steiner Tree} without costs.

Assume now that $(G,S,\ell+r)$ is a no-instance of \textsc{Secluded Steiner Tree} without costs. 
Let $T$ be a Steiner tree for the set of terminals $S$. Because for any connected subgraph $T'$ of $G$ such that $S\subseteq V(T')$, $N_G[V(T')]=V(G)$, and because $(G,S,\ell+r)$ is a no-instance, 
$\ell=|V(T)|<|V(G)|-r=(n+m+1)-(m-k-1)=n+k+2$. Let $\{x_{i_1},\ldots,x_{i_{k'}}\}=\{x_1,\ldots,x_k\}\cap V(T)$. Since $|V(T)|\leq n+k+1$, we obtain that $k'\leq k$. It remains to note that $X_{i_1},\ldots,X_{i_{k'}}$ cover $U$ and, therefore, $(U,X_1,\ldots,X_m,k)$ is a yes-instance of \textsc{Set Cover}.
\end{proof}

\section{Structural parameterizations of  Secluded \\
Steiner Tree}\label{sec:struct}

In this section we consider different algorithmic and complexity results concerning different structural parameterizations of secluded connectivity problems. We consider   parameterizations by  the treewidth, size of the solution, maximum degree and the size of a vertex cover of the input graph. (See Appendix for definitions of these parameters.) We show that
  it is unlikely that \textsc{Secluded Path} without costs parameterized by $k$, the treewidth and the maximum degree of the input graph  
has a polynomial kernel. We obtain the same result for the cases when the problem is parameterized by $k-\ell$, the treewidth and the maximum degree of the input graph, where $\ell$ is the length of the shortest path between terminals.

\begin{theorem}\label{thm:no-kern-tw}
\textsc{Secluded Path} without costs on graphs of treewidth at most $t$ and maximum degree at most $\Delta$ admits no polynomial kernel unless  $\classNP\subseteq\classCoNP/\text{\rm poly}$ when parameterized by $k+t+\Delta$ or $(k-\ell)+t+\Delta$, where $\ell$ is the length of the shortest path between terminals.
\end{theorem}

The proof uses the cross-composition technique introduced by Bodlaender, Jansen and Kratsch~\cite{BodlaenderJK14}.
We need the following additional definitions (see~\cite{BodlaenderJK14}).

Let $\Sigma$ be a finite alphabet. An equivalence relation $\mathcal{R}$ on the set of strings $\Sigma^*$ is called a \emph{polynomial equivalence relation} if the following two conditions hold:
\begin{itemize}
\item[i)] there is an algorithm that given two strings $x,y\in\Sigma^*$ decides whether $x$ and $y$ belong to
the same equivalence class in time polynomial in $|x|+|y|$,
\item[ii)] for any finite set $S\subseteq\Sigma^*$, the equivalence relation $\mathcal{R}$ partitions the elements of $S$ into a
number of classes that is polynomially bounded in the size of the largest element of $S$.
\end{itemize}

Let $L\subseteq\Sigma^*$ be a language, let $\mathcal{R}$ be a polynomial
equivalence relation on $\Sigma^*$, and let $\mathcal{Q}\subseteq\Sigma^*\times\mathbb{N}$   
be a parameterized problem.  An \emph{OR-cross-composition of $L$ into $\mathcal{Q}$} (with respect to $\mathcal{R}$) is an algorithm that, given $t$ instances $x_1,x_2,\ldots,x_t\in\Sigma^*$ 
of $L$ belonging to the same equivalence class of $\mathcal{R}$, takes time polynomial in
$\sum_{i=1}^t|x_i|$ and outputs an instance $(y,k)\in \Sigma^*\times \mathbb{N}$ such that:
\begin{itemize}
\item[i)] the parameter value $k$ is polynomially bounded in $\max\{|x_1|,\ldots,|x_t|\} + \log t$,
\item[ii)] the instance $(y,k)$ is a yes-instance for $\mathcal{Q}$ if and only if at least one instance $x_i$ is a yes-instance for $L$ for $i\in\{1,\ldots,t\}$.
\end{itemize}
It is said that $L$ \emph{OR-cross-composes into} $\mathcal{Q}$ if a cross-composition
algorithm exists for a suitable relation $\mathcal{R}$.

In particular, Bodlaender, Jansen and Kratsch~\cite{BodlaenderJK14} proved the following theorem.

\begin{theorem}[\cite{BodlaenderJK14}]\label{thm:BJK}
If an \classNP-hard language $L$ OR-cross-composes into the parameterized problem $\mathcal{Q}$,
then $\mathcal{Q}$ does not admit a polynomial kernelization unless
$\classNP\subseteq\classCoNP/\text{\rm poly}$.
\end{theorem}

\begin{proof}[of Theorem~\ref{thm:no-kern-tw}]
First, we prove the claim for the case when the problem is parameterized by $k+t+\Delta$.

We construct an OR-composition of \textsc{Secluded Path} without costs to  the parameterized version of  \textsc{Secluded Path}. Recall that \textsc{Secluded Path} without costs  was shown to be \classNP-complete by Chechik et al,~\cite{ChechikJPP12,ChechikJPP13}.
We assume that two instances $(G,\{s_1,s_2\},k)$ and $(G',\{s_1',s_2'\},k')$ of \textsc{Secluded Path} without costs  are equivalent if $|V(G)|=|V(G')|$ and $k=k'$. 
Let $(G_i,\{s_1^i,s_2^i\},k)$ for $i\in\{1,\ldots,p\}$ be  equivalent instances of \textsc{Secluded Path}, $|V(G_i)|=n\geq 3$. 
Without loss of generality we assume that $p=2^q$ for a positive integer $q$; otherwise, we add minimum number of copies of $(G_1,\{s_1^1,s_2^1\},k)$ to achieve this property. We construct the graph $G$ as follows.  
\begin{itemize}
\item[i)] Construct disjoint copies of $G_1,\ldots,G_p$.
\item[ii)] Construct a rooted binary tree $T_1$ of height $q$, denote the root by $s_1$ and identify $t=2^q$ leaves of the tree with the vertices of $s_1^1,\ldots,s_1^p$ of $G_1,\ldots,G_p$.
\item[iii)] Construct a rooted binary tree $T_2$ of height $q$, denote the root by $s_2$ and identify $t=2^q$ leaves of the tree with the vertices of $s_2^1,\ldots,s_2^p$ of $G_1,\ldots,G_p$.
\end{itemize}
We set $k'=k+4q$ and consider the instance $(G,\{s_1,s_2\},k')$ of  \textsc{Secluded Path}. Notice that $\tw(G_i)\leq n-1$ and $\Delta(G_i)\leq n-1$ for $i\in\{1,\ldots,p\}$ and $\tw(G)\leq n-1$ and $\Delta(G)\leq n$.

We claim that $G$ has an $(s_1,s_2)$-path $P$ with $|N_G[V(P)]|\leq k'$ if and only if there is $i\in\{1,\ldots,p\}$ such that $G_i$ has an $(s_1^i,s_2^i)$-path $P_i$ with $|N_{G_i}[V(P_i)]|\leq k$.

Let $P$ be an $(s_1,s_2)$-path $P$ in $G$ with $|N_G[V(P)]|\leq k'$. Consider the first vertex $u$ of $P$ starting from $s_1$ that is a leaf of $T_1$. Clearly,  $u\in\{s_1^i,s_2^i\}$ for some $i\in\{1,\ldots,p\}$. Without loss of generality we can assume that $u=s_1$. Notice that $P$ contains $s_2^i$ by the construction of $G$ and the $(s_1,s_2)$-subpath $P_i$ of $P$ is an $(s_1,s_2)$-path in $G_i$. It remains to observe that  $k'\geq |N_G[V(P)]|\geq 4q+|N_{G_i}[V(P_i)]|$ and, therefore, $|N_{G_i}[V(P_i)]|\leq k$.

Suppose that $G_i$ has an $(s_1^i,s_2^i)$-path $P_i$ with $|N_{G_i}[V(P_i)]|\leq k$ for some $i\in\{1,\ldots,p\}$. Let $P'$ be the unique $(s_1,s_1^i)$-path in $T_1$ and let $P''$ be the unique $(s_2^i,s_2)$-path in $T_2$. We have that for the $(s_1,s_2)$-path $P$ in $G$ obtained by the concatenation of $P'$, $P_i$ in the copy of $G_i$ and $P''$,
$|N_G[V(P)]|\leq k+4q=k'$. 

The proof for the case when the problem is parameterized by $(k-\ell)+t+\Delta$ uses the same  OR-composition. The difference is that now we assume that 
 two instances $(G,\{s_1,s_2\},k)$ and $(G',\{s_1',s_2'\},k')$ are equivalent if $|V(G)|=|V(G')|$, $k=k'$ and $s_1,s_2$ and $s_1\rq{},s_2\rq{}$ are at the same distance in $G$ and $G\rq{}$ respectively. Let $\ell$ be the distance between $s_1^i$ and $s_2^i$ in $G_i$ for   $i\in\{1,\ldots,p\}$. Then the length of a shortest $(s_1,s_2)$-path in $G$ is $\ell\rq{}=\ell+2q$. 
Hence $k\rq{}-\ell\rq{}=k-\ell+2q$.
\end{proof}

Observe that Theorem~\ref{thm:no-kern-tw} immediately implies that \textsc{Secluded Path} without costs has no polynomial kernel unless  $\classNP\subseteq\classCoNP/\text{\rm poly}$ when parameterized by $k$ or $k-\ell$.
The next natural question is if parameterization by a stronger parameter can lead to a polynomial kernel.
Let us note that  the treewidth of a graph is always at most the minimum size of its vertex cover.
The following theorem provides lower bounds for parameterization by the minimum size of a vertex cover. 

\begin{theorem}\label{thm:no-kern-vc}
\textsc{Secluded Path} without costs on graphs with the vertex cover number at most $w$ has no polynomial kernel unless  $\classNP\subseteq\classCoNP/\text{\rm poly}$ when parameterized by $w$.
\end{theorem}

\begin{proof}
We show that the {\sc $3$-Satisfiability} problem OR-cross composes into \textsc{Secluded Path} without costs. 
Recall that \textsc{$3$-Satisfiability} asks for given boolean variables $x_1,\ldots,x_n$ and clauses $C_1,\ldots,C_m$ with 3 literals each, whether the formula $\phi=C_1\wedge\ldots\wedge C_m$ can be satisfied.
It is well-known that  \textsc{$3$-Satisfiability} is \classNP-complete~\cite{GareyJ79}.
We assume that two instances of \textsc{$3$-Satisfiability} are equivalent if they have the same number of variables and the same number of clauses.  

\begin{figure}[ht]
\centering\scalebox{0.7}{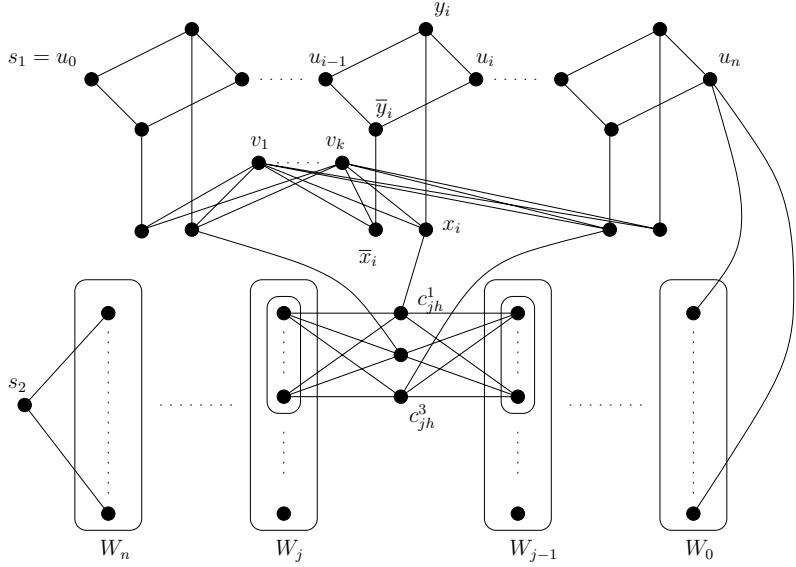}
\caption{Construction of $G$.
\label{fig:G}}
\end{figure}

Consider $t$ equivalent instances of \textsc{$3$-Satisfiability} with the same boolean variables $x_1,\ldots,x_n$ and the sets of clauses $\mathcal{C}_i=\{C_1^i,\ldots,C_m^i\}$ for $i\in\{1,\ldots,t\}$. 
Without loss of generality we assume that $t=\binom{2q}{q}$ for a positive integer $q$; otherwise, we add minimum number of copies of $\mathcal{C}_1$ to get this property. Notice that $\binom{2q}{q}=\Theta(4^q/\sqrt{\pi q})$ and $q=O(\log t)$. Let $I_1,\ldots,I_t$ be pairwise distinct subsets of $\{1,\ldots,2q\}$ of size $q$. Notice that each $i\in\{1,\ldots,2q\}$ is included exactly in $d=\binom{2q-1}{q-1}$ sets. Let $k=(q+3d)m+3q+4n+2$.
We construct the graph $G$ as follows (see Fig.~\ref{fig:G}).
\begin{itemize}
\item[i)] Construct $n+1$ vertices $u_0,\ldots,v_n$. Let $s_1=u_0$.
\item[ii)] For each $i\in \{1,\ldots,n\}$, construct vertices $x_i,y_i,\overline{x}_i,\overline{y}_i$ and edges $u_{i-1}y_i,y_iu_i,y_ix_i,$ and $u_{i-1}\overline{y}_i,\overline{y}_iu_i,\overline{y}_i\overline{x}_i$.  
\item[iii)] For each $j\in\{0,\ldots,m\}$, construct a set of vertices\\ $W_j=\{w_1^j,\ldots,w_{2q}^j\}$.
\item[iv)] Construct a vertex $s_2$ and edges $u_2w_1^0,\ldots,u_2w_{2q}^0$  and  $w_1^ms_2,\ldots,w_{2q}^ms_2$.
\item[v)] For each $j\in\{1,\ldots,m\}$ and $h\in\{1,\ldots,t\}$, 
\begin{itemize}
\item construct 3 vertices $c_{jh}^1,c_{jh}^2,c_{jh}^3$;
\item construct edges $c_{jh}^1w_r^{j-1},c_{jh}^2w_r^{j-1},c_{jh}^3w_r^{j-1}$ and\\ $c_{jh}^1w_r^{j},c_{jh}^2w_r^{j},c_{jh}^3w_r^{j}$ for all $r\in I_h$;
\item consider the clause $C_j^h=(z_1\vee z_2\vee z_3)$ and for $l\in\{1,2,3\}$, construct an edge $c_{jh}^lx_i$ 
if $z_l=x_i$ for some $i\in\{1,\ldots,n\}$ and construct an edge $c_{jh}^l\overline{x}_i$
 if $z_l=\overline{x}_i$.  
\end{itemize}
\item[vi)] Construct $k$ vertices $v_1,\ldots,v_k$ and edges $x_iv_l,\overline{x}_iv_l$ for $i\in\{1,\ldots,n\}$ and $l\in\{1,\ldots,k\}$.
\end{itemize}
Observe that the set of vertices
$$X=(\cup_{i=1}^n\{x_i,y_i,\overline{x}_i,\overline{y}_i\})\cup(\cup_{j=0}^m W_j)$$
is a vertex cover in $G$ of size $4n+2q(m+1)=O(n+m\log t)$.

We show that $G$ has an $(s_1,s_2)$-path $P$ with $|N_G[V(P)]|\leq k$ if and only if there is $h\in\{1,\ldots,t\}$ such that $x_1,\ldots,x_n$ have a truth assignment satisfying all the clauses of $\mathcal{C}_h$.

Suppose that $x_1,\ldots,x_n$ have an assignment that satisfies all the clauses of $\mathcal{C}_h$. First, we construct the $(s_1,u_n)$-path $P'$ by the concatenation of the following paths: for each $i\in\{1,\ldots,n\}$, we take the path $u_{i-1}y_iu_i$ if $x_i=true$ in the assignment and we take $u_{i-1}\overline{y}_iu_i$ if $x_i=false$. 
Let $r\in I_h$. We construct the $(w_r^0,w_r^m)$-path $P''$ by concatenating $w_r^{j-1}c_{jh}^{l_j}w_r^j$ for $j\in\{1,\ldots,m\}$ where $l_j\in\{1,2,3\}$ is chosen as follows. Each clause $C_j^h=z_1\vee z_2\vee z_3=true$ for the assignment, i.e., $z_l=true$ for some $l\in\{1,2,3\}$; we set $l_j=l$.   
Finally, we set $P=P'+u_nw_h^0+P''+w_h^ms_2$.
It is straightforward to verify that $|N_G[V(P)]|=k$.

Suppose now that there is an $(s_1,s_2)$-path in $G$ with $|N_G[V(P)]|\leq k$. We assume that $P$ is an induced path.
Observe that $x_i,\overline{x_i}\notin V(P)$ for $i\in\{1,\ldots,n\}$, because $d_G(x_i),d_G(\overline{x}_i)>k$.
Therefore, $P$ has an $(s_1,u_n)$-subpath $P'$ such that $u_0,\ldots,u_n\in V(P')$ and for each $i\in\{1,\ldots,n\}$, either $y_i\in V(P')$ or $\overline{y}_i\in V(P')$. We set the variable $x_i=true$ if $x_i\in V(P')$ and $x_i=false$ otherwise. We show that this truth assignment satisfies all the clauses of some $\mathcal{C}_r$.

Observe that $|N_G[V(P')]|=4n+2q+1$. Clearly, $s_2\in V(P)$. Notice also that $P$ has at least one vertex in each $W_j$ for $j\in\{0,\ldots,m\}$, and for each $j\in\{1,\ldots,m\}$, at least one vertex among the vertices $c_{jh}^l$ for $h\in\{1,\ldots,t\}$ and $l\in\{1,2,3\}$ is in $P$. For each $j\in\{1,\ldots,m\}$, any two verices $w_r^{j-1}\in W_{j-1}$ and $w_{r'}^j\in W_j$ have at least $3d$ neighbors among the vertices $c_{jf}^l$ for $f\in\{1,\ldots,t\}$ and $l\in\{1,2,3\}$. Moreover, if $r\neq r'$, they have at least $3d+6$ such neighbors, because there are two subsets $I,I'\subseteq\{1,\ldots,2q\}$ of size $q$ such that $r\in I\setminus I'$ and $r'\in I'\setminus I$. For each $j\in \{1,\ldots,m-1\}$, any two vertices $c_{jh}^l$ and $c_{j+1~h'}^{l'}$ for $h,h'\in\{1,\ldots,t\}$ and $l,l'\in\{1,2,3\}$ have at least $q$ neighbors in $W_j$. Moreover, if $h\neq h'$, they have at least $q+2$ such neighbors, because $|I_h\cup I_{h'}|\geq q+2$.  Taking into account that $d_G(s_2)=2q$, we obtain that 
$$k\geq |N_G[V(P)]|\geq |N_G[V(P')]|+3dm+q(m-1)+2q+1=k.$$
It implies that 
$P$ has exactly one vertex in each $W_j$ for $j\in\{0,\ldots,m\}$, and 
for each $j\in\{1,\ldots,m\}$, exactly one vertex among the vertices $c_{jh}^l$ for $h\in\{1,\ldots,t\}$ and $l\in\{1,2,3\}$ is in $P$. Moreover, there is $r\in\{1,\ldots,2q\}$ and $h\in\{1,\ldots,t\}$ such that $w_r^j\in V(P)$ and $c_{jh}^{l_j}\in V(P)$ for $j\in\{0,\ldots,m\}$ and $l_j\in\{1,2,3\}$. We claim that all the clauses of $\mathcal{C}_r$ are satisfied. Otherwise, if there is a clause $C_j^r=(z_1\vee z_2\vee z_3)$ that is not satisfied, then the neighbors of $c_{jh}^1,c_{jh}^2,c_{jh}^3$ among the vertices $x_i,\overline{x}_i$ for $i\in\{1,\ldots,n\}$ are not in $N_G[V(P')]$. It immediately implies that $|N_G[V(P)]|>k$; a contradiction.
\qed
\end{proof}

However, if we consider even stronger parameterization, by vertex cover number and by the size of the solution, then we obtain the following theorem.

\begin{theorem}\label{thm:kern-vc}
The \textsc{Secluded Steiner Tree} problem admits a kernel with at most $2w(k+1)$ vertices on graphs with the vertex cover number at most $w$. 
\end{theorem}

\begin{proof}
Let $(G,S,k)$ be an instance of  \textsc{Secluded Steiner Tree}. We assume that $|S|\geq 2$, as otherwise the problem is trivial.
Our kernelization algorithm uses the following steps.

\medskip
\noindent
{\bf Step 1.} If $G$ is disconnected, then return a no-answer and stop if there are distinct components of $G$ that contain terminals, and construct the instance $(G',S,k)$ if there is a component $G'$ of $G$ with $S\subseteq V(G)$.

\medskip
It is straightforward to see that our first step is safe to apply, i.e., it either returns a correct answer or creates an equivalent instance of our problem. From now we assume that $G$ is connected.

\medskip
\noindent
{\bf Step 2.} Find a set of vertices $X$ by taking end-vertices of the edges of a maximal matching in $G$. If $|X|>2w$, then return a no-answer and stop.

\medskip
It is well-known (see e.g.~\cite{GareyJ79}) that $X$ is a vertex cover and $|X|$ gives a factor-2 approximation of the vertex cover number. In particular, if $|X|>2w$, then $G$ has no vertex cover of size at most $w$. 

\medskip
\noindent
{\bf Step 3.} Let $Y=\{v\in X|d_G(v)\leq k\}$, $I=N_G(Y)\setminus X$ and $I'=V(G)\setminus(X\cup N_G(Y))$. If $S\cap( X\setminus Y)\neq\emptyset$ or $S\cap I'\neq\emptyset$, then return a no-answer and stop.
  
\medskip
Clearly, if $T$ is a connected subgraph of $G$ with $S\subseteq V(T)$ such that $|N_G[V(T)]|\leq k$, then $V(T)\cap (X\setminus Y)=\emptyset$. We also have that $V(T)\cap I'=\emptyset$. To see it, assume that $u\in V(T)\cap I'$. Since $|S|\geq 2$, $T$ has no isolated vertices and, therefore, $u$ has a neighbor $v$ in $T$, but then $v\in X\setminus Y$; a contradiction.   
It proves that Step~3 is safe.

\medskip
\noindent
{\bf Step 4.} Delete the vertices of $I'$ and $X\setminus N_G[Y\cup I]$. 
If $I'\neq\emptyset$, then add $k$ vertices of cost $0$ and make them adjacent to the vertics of $N_G(Y\cup I)\cap X$.

\medskip
Denote by $G'$ the graph obtained on Step~4. If $T$ is a connected subgraph of $G$ such that $S\subseteq V(T)$ and $|N_G[V(T)]|\leq k$, then $T$ is a subgraph of $G'$ and $N_{G'}[V(T)]=N_G[V(T)]$, because $V(T)\cap (X\setminus Y)=\emptyset$ and $V(T)\cap I'=\emptyset$.  Suppose that $T$ is a connected subgraph of $G'$ with $S\subseteq V(T')$ such that $|N_{G'}[V(P)]|\leq k$. Then $T$ does not contain any added vertex, because they are adjacent only to the vertices of degree at least $k+1$,  and $V(T')\cap (X\setminus Y)=\emptyset$. Hence,  
 $T$ is a subgraph of $G$ and $N_{G'}[V(T)]=N_G[V(T)]$.

Now we give an upper bound for  the size of $G'$.
If $I'=\emptyset$, then $V(G)\setminus X\subseteq N_G(Y)$ and $|V(G')|\leq 2w(k+1)$. If $I'\neq\emptyset$, then $X\setminus X\neq\emptyset$ and, therefore, $|V(G')|\leq |X|+|Y|k+k\leq |X|(k+1)\leq 2w(k+1)$.

It is straightforward to see that Steps~1--4 can be done in polynomial time and it concludes the proof.
\end{proof}

Recall that Chechik et al. \cite{ChechikJPP13}  showed that if  the treewidth of the input graph does not exceed $t$, then  the \textsc{Secluded Steiner Tree} problem is solvable in time 
$2^{\cO(t\log{t})} \cdot n^{O(1)}\cdot \log{\cost}$, 
where $\cost$ is the maximum value of $\costfunction$ on  an input graph $G$. 
We observe that the running time could be improved by applying modern techniques for dynamic programming over tree decompositions proposed by Cygan et al. \cite{cut-and-count}, 
Bodlaender et al.~\cite{BodlaenderCKN13} and Fomin et al.~\cite{fomin2014efficient}. Essentially, the algorithms for \textsc{Secluded Steiner Tree}
are constructed along the same lines as the algorithms for \textsc{Steiner Tree} described in~\cite{cut-and-count,BodlaenderCKN13,fomin2014efficient}. Hence, for simplicity, we only sketch the randomized algorithm based on the Cut\&Count technique introduced by   Cygan et al.~\cite{cut-and-count} for \textsc{Secluded Steiner Tree} without costs in this conference version of our paper. 

\begin{theorem}\label{thm:CandC}
There is a 
true-biased
Monte Carlo algorithm 
solving the
\textsc{Secluded Steiner Tree} without costs in time $4^{t}\cdot n^{\cO(1)}$, given a tree decomposition of width at most $t$.
\end{theorem}

We need some additional definitions and auxiliary results.

Let $(\mathcal{B},T)$ be a tree decomposition of a graph $G$, $\mathcal{B}=\{B_i\mid i\in V(T)\}$.
We distinguish one vertex $r$ of $T$ which is said to be a \emph{root} of $T$. This introduces natural
parent-child and ancestor-descendant relations in the tree $T$. 
We  say that  a rooted tree decomposition $(\mathcal{B},T)$ is an \emph{extended nice} tree decomposition if the following
conditions are satisfied:
\begin{itemize}
\item $X_r=\emptyset$ and $X_\ell=\emptyset$ for every leaf $\ell$ of $T$. In other words, all the leaves as
well as the root contain empty bags.
\item For every edge $uv\in E(G)$, there is the unique bag $B_i$ assigned to $uv$ such that $u,v\in B_i$; we say that this bag \emph{is labeled} by $uv$.
\item Every non-leaf node of $T$ is of one of the following three types:
\begin{itemize}
\item {\bf Introduce vertex node:} a node $h$ with exactly one child $h\rq{}$ such that $B_h =B_{h\rq{}}\cup\{v\}$ for some vertex $v\notin B_{h\rq{}}$; we say that $v$ \emph{is introduced} at $h$.
\item {\bf Introduce edge node:} a node $h$ labeled with an edge $uv\in E(G)$ such that $u,v\in B_h$, and with exactly one child $h\rq{}$ such that $B_h=B_{h\rq{}}$. We say
that edge $uv$ \emph{is introduced} at $h$. 
\item {\bf Forget node:} a node $h$ with exactly one child $h\rq{}$ such that $B_h=B_{h\rq{}}\setminus\{w\}$ for some vertex $w\in B_{h\rq{}}$; we say that $w$ \emph{is forgotten} at $h$.
\item {\bf Join node:} a node $h$ with exactly two children $h_1$ and  $h_2$ such that $B_h= B_{h_1}=B_{h_2}$ .
\end{itemize}
\item All the edges incident to a vertex $v\in V(G)$ are introduced immediately after $v$ is introduced.
\end{itemize}
Using the same arguments as in~\cite{Kloks94}, it is straightforward to show that for a given tree decomposition $(X,T)$ of a graph $G$ of width $t$, an extended nice tree decomposition of $G$ of  width at most $t$ such that the total size of the obtained tree is $O(t^2|V(T)|)$ can be constructed in linear time.

For a function $w \colon U \to \mathbb{Z}$ and a set $S \subseteq U$, let $w(S)=\sum_{u \in S}w(u)$. We say that $w$ \emph{isolates} a set family $\mathcal{F} \in 2^{U}$ if there is a unique $S' \in \mathcal{F'}$ satisfying $w(S')=\min_{S \in \mathcal{F}}w(S)$.
The Cut\&Count approach uses the following statement proved by Mulmuley et al.~\cite{MulmuleyVV87}.

\begin{lemma}[Isolation Lemma, \cite{MulmuleyVV87}]\label{isolation_lemma}
Let $\mathcal{F} \subseteq 2^{U}$ be a set family over a universe $U$ with $|\mathcal{F}| > 0$. 
For each $u \in U$ choose a weight $w(u) \in \{1, \ldots, N\}$ uniformly and independently at
random. Then
 \[ \operatorname{Pr}(w \text{ isolates } \mathcal{F}) \geq 1 - \frac{|U|}{N} \]
\end{lemma}

\begin{proof}[of Theorem~\ref{thm:CandC}]
We will search for a subset of vertices $X \subseteq V(G)$
such that
\begin{equation}
\text{$S \subseteq X$, $G[X]$ is connected, and $|N_G[X]| \le k$.}
\label{eq:f1}
\end{equation}
It is not difficult to see that such a set $X$ exists if and only if there exists a pair $(X,Y)$ of disjoint sets such that
\begin{equation}
\text{$S \subseteq X$, $G[X]$ is connected, $N_G[X] \subseteq X \cup Y$, and $|X|+|Y| \le k$}
\label{eq:f2}
\end{equation}
(for this, take $Y=N_G[X] \setminus X$).
We use the standard dynamic programming on tree decompositions together with the cut and count technique.

Assume that each vertex $v \in V$ is assigned an integer weight~$w(v)$.
To use dynamic programming we  relax the restriction that $G[X]$ is connected. Namely, we view $X$ as a union of two disjoint sets $X_0$ and $X_1$ between them. Let $\mathcal{R}_{w, s}$ be the set of all disjoint triples $(X_0, X_1, Y)$ such that
\begin{multline}
\text{$S \subseteq X_0 \cup X_1$, $N_G[X_0] \cap X_1 = \emptyset$, $N_G[X_0 \cup X_1] \subseteq X_0 \cup X_1 \cup Y$,} \\ \text{$w(X_0 \cup X_1) = w$, and
$|X_0 \cup X_1 \cup Y| = s$.}
\label{eq:f3}
\end{multline}
Note that any pair $(X,Y)$ satisfying \eqref{eq:f2} such that $G[X]$ consists of $l$ connected components, contributes exactly $2^l$ triples to $\mathcal{R}_{w,s}$
(just because each of the $l$ connected components can go to either $X_0$ or $X_1$). Hence if we compute $|\mathcal{R}_{w,s}|$ modulo $4$ all pairs $(X,Y)$ with disconnected $X$ will cancel out.

Let now $s'$ be the minimum possible integer such that there exists $X \subseteq V$ with $|N_G[X]|=s'$ satisfying~\eqref{eq:f1}. 
Consider a set family $\mathcal{F} \subseteq 2^{V(G)}$ consisting of all such sets $X$ 
(i.e., $X$ satisfies~\eqref{eq:f1} and $N_G[x]=s'$).
Lemma~\ref{isolation_lemma} guarantees that if each vertex $v \in V(G)$ is assigned a random weight from 
$\{1,\ldots,2n\}$
then
$\mathcal{F}$ contains a unique set $X$ such that $w(X)=w'$
where $w'=\min_{S \in \mathcal{F}}w(S)$ with probability at least~$1/2$. This in turn implies that $|\mathcal{R}_{w',s'}| \equiv 2 \pmod{4}$ with probability at least $1/2$. This allows us to conclude that with probability at least $1/2$ we will find $s'$ by computing 
$|\mathcal{R}_{w,s}| \bmod{4}$ for all $w$ and~$s$.
We turn to show how to compute this.

Recall that we are given a tree decomposition $T$ of $G$ of width $t$.
Without loss of generality assume that the given tree decomposition is an extended nice decomposition. For a vertex $h \in V(T)$,
let $B_h \subseteq V(G)$ be its bag, 
$V_h \subseteq V(G)$ and $E_h\subseteq E(G)$ be all the vertices end edges of $G$ respectively that are introduced in the subtree of $T$ rooted at $h$, and
$G_h$ be a graph on the vertex set $V_h$ containing all
the edges introduced in that subtree.  

By a \emph{coloring} of a bag $B_h$ we mean
a mapping $f \colon B_h \to \{0_0, 0_1, 1_0, 1_1\}$ assigning four different colors to the vertices
of the bag.
\begin{itemize}
\item
\textbf{Red}, represented by $1_0$. The meaning is that all red vertices have to be 
contained in~$X_0$.
\item
\textbf{Blue}, represented by $1_1$. The meaning is that all blue vertices have to be
conatained in~$X_1$.
\item
\textbf{Green}, represented by $0_1$. The meaning is that all green vertices have to be
contained in~$Y$.
\item
\textbf{White}, represented by $0_0$. The meaning is that all white vertices do not appear in $X_0 \cup X_1 \cup Y$. 
\end{itemize}

Given a coloring $f$ of a bag $B_h$, we say that a triple $(P_0, P_1, Q)$ of pairwise disjoint subsets of $V_h$ is \emph{nice} with respect to $t$ and $f$ if 
\begin{itemize}
\item all the vertices from $B_t$ are colored properly:
\begin{equation}
f^{-1}(1_0) = P_0 \cap B_t, \, 
f^{-1}(1_1) = P_1 \cap B_t, \, 
f^{-1}(0_1) = Q \cap B_t \, .
\label{eq:n1}
\end{equation}
\item there are no edges between vertices from $P_0$ and $P_1$ in $G_t$:
\begin{equation}
\text{$uv\not \in E_h$ for $u \in P_0, v \in P_1$}
\label{eq:n2}
\end{equation}
\item any neighbor of a vertex from $P_0 \cup P_1$ lies in
$P_0 \cup P_1 \cup Q$: 
\begin{equation}
\text{if $u \in P_0 \cup P_1$ and $uv \in E_h$ then $v \in P_0 \cup P_1 \cup Q$.}
\label{eq:n3}
\end{equation}
\end{itemize}
Accordingly, the \emph{size} of the triple $(P_0, P_1, Q)$
is $|P_0 \cup P_1 \cup Q|$ and its \emph{weight}
is $w(P_0 \cup P_1)$.

We are now ready to define a state of our dynamic programming algorithm: $c[h,f,s,w]$ is the number modulo $4$ of nice triples of size $s$ and weight $w$ with respect to $h$ and~$f$. Clearly, the number of states is $\cO(t \cdot 4^{t} \cdot n^3)$ (since $s$ is at most $n$ and $w$ is at most $4n^2$). Below we show how to compute all the states by going through the given tree decomposition from the leaves to the root.

\textbf{Leaf node}. If $h$ if a leaf node then $B_h = \emptyset$. Then the only possible coloring is just the empty coloring and the only nice triple with respect to $h$ and this empty coloring is $(\emptyset, \emptyset, \emptyset)$. Hence for all $s,w$, 
\[c[h, \emptyset, s,w]=[w = 0 \land s = 0] \, .\] 

\textbf{Introduce vertex node}. Let $h$ be an introduce node and $h'$ be its child such that
$X_h = X_{h'} \cup \{v\}$ for some $v \not\in B_{h'}$.
Note that $v$ is an isolated vertex in~$G_h$.
If $v$ is not a terminal vertex (i.e., $v \not\in S$)
it can be colored using any of our four colors. 
While if $v$ is a terminal vertex it should be colored either red or blue. We arrive at the following formula where each case is applied only if none of the previous cases is applicable:
\[ c[h, f_{v \to \alpha}, s, w] = 
\begin{cases}
c[h', f, s - 1, w - w(v)] & 
\text{if $\alpha = 1_0 \lor \alpha = 1_1$} \\
0 & \text{if $v \in S$} \\
c[h', f, s - 1, w] & \text{if $\alpha = 0_1$} \\
c[h', f, s, w] & \text{if $\alpha = 0_0$}
\end{cases}
\]

\textbf{Introduce edge node}. Let $h$ be an introduce edge $uv$ with a child $t'$ such that
$E_h = E_{h'} \cup \{uv\}$ for some $u, v \in B_{h'}$ and $f$ be a coloring of $B_h$.
Clearly any triple that is nice with respect to $h$ and $f$ is also nice with respect to $h'$ and~$f$.
Hence all we need to do is to check whether all constraints are satisfied for the new edge $uv$. I.e., this edge should not join a blue vertex with a red one or
a blue/red vertex with a white one. Formally,
\[ 
c[t, f, s, w] = 
\begin{cases}
c[h', f, s, w] & \text{if $\{f(u),f(v)\} \in \{\{1_0,1_1\}, \{1_0,0_0\}, \{1_1,0_0\} \}$},\\
0 & \text{otherwise.}
\end{cases}
\]

\textbf{Forget node}. Let $h$ be a forget node with a child $h'$ such that $X_h = X_{h\rq{}} \setminus \{v\}$ for some 
$v \in B_{h'}$. 
Then clearly
\[ c[h, f, s, w] = 
\left(\sum\limits_{\alpha \in \{1_0, 1_1, 0_0, 0_1\}}c[h', f_{v \to \alpha}, s, w]\right) \bmod{4} \, .\]

\textbf{Join node}. Let $t$ be a join node with children $h_1$ and $h_2$ such that
$B_h = B_{h_1} = B_{h_2}$. 
Let $f$ be a coloring of $B_h$ (and hence also a coloring
of $B_{h_1}$ and $B_{h_2}$). Note that there is a natural one-to-one correspondence between nice triples for $h,f$ and $f$ and pairs on nice triples for $h_1,f$ and $h_2,f$. Namely, a nice triple $(P_0,P_1,Q)$ for $t,f$ defines a nice triple $(P_0^1,P_1^1,Q^1)$ for $h_1,f$ and a nice triple $(P_0^2,P_1^2,Q^2)$ for $h_2,f$ as follows ($i=1,2$):
\[P_0^i=P_0 \cap V_{h_i}, \, P_1^i=P_1 \cap V_{h_i}, \, Q^i=Q \cap V_{h_i} \, .\] And vice versa, two nice triples 
$(P_0^1,P_1^1,Q^1)$ and $(P_0^2,P_1^2,Q^2)$ define a nice triple $(P_0,P_1,Q)$ as follows:
\[P_0 = P^{1}_0 \cup P^{2}_0, \, P_1 = P^{1}_1 \cup P^{2}_1,\,
Q = Q^{1} \cup Q^{2} \, .\]
It is straightforward to check that the properties \eqref{eq:n1}--\eqref{eq:n3} are satisfied for both these maps.
This allows us to use the following formula for computing the current state. Let $s(f) = |f^{-1}(0_1) \cup f^{-1}(1_0) \cup f^{-1}(1_1)|$ and $w(f) = w(f^{-1}(1_0) \cup f^{-1}(1_1))$. Then
\[ 
c[h, f, s, w] = \left( \sum_{\substack{s_1 + s_2 = s + s(f) \\ w_1 + w_2 = w + w(f)}}c[h_1, f, s_1, w_1]{\cdot}c[h_2, f, s_2, w_2] \right) \bmod{4}
\]

This finishes the description of the dynamic programming algorithm for filling in the table~$c[]$. From this table one can easily extract the value of $\mathcal{R}_{w,s} \bmod{4}$: it is just $c[r, \emptyset, s, w]$ where $r$ is the root node of the given tree decomposition.

To conclude, it remains to note that each node in the given tree decomposition is processed in time $4^{t}\cdot n^{\cO(1)}$.
 \end{proof}

The algorithm based on the Cut\&Count technique can be generalized for \textsc{Secluded Steiner Tree} with costs in the same way as the algorithm for \textsc{Steiner Tree} in~\cite{cut-and-count}.  This way  we can obtain the algorithm that runs in time $4^{t}\cdot(n+W)^{\cO(1)}$ where $W$ is the maximal cost of vertices.
One can obtain a deterministic algorithm and improve the dependence on $W$ using the representative set technique for dynamic programming over tree decompositions introduced by Fomin et al.~\cite{fomin2014efficient}. Again by the same approach as for \textsc{Steiner Tree}, it is possible to  solve \textsc{Secluded Steiner Tree} deterministically in time $\cO((2+2^{\omega+1})^{t}\cdot(n+\log{W})^{\cO(1)})$ (here $\omega$ is the matrix multiplication constant).

\end{document}